\documentclass[11pt]{article}
\usepackage{authblk}
\usepackage{amsmath}
\usepackage{amsfonts}
\usepackage{amssymb}
\usepackage{amsthm}
\usepackage{fullpage}
\usepackage{graphicx}
\usepackage{enumerate}
\def\be{\begin{equation}}
\def\ee{\end{equation}}
\def\bea{\begin{eqnarray}}
\def\eea{\end{eqnarray}}
\def\bma{\begin{mathletters}}
\def\ema{\end{mathletters}}

\def\0{\overline{0}}
\def\tr{\mbox{tr}}

\def\q0{\underline{0}}

\def\H{{\cal H}}

\def\id{{\mathbb I}}

\def\H{{\cal H}}

\def\tr{\mbox{tr}}
\def\one{\leavevmode\hbox{\small1\normalsize\kern-.33em1}}

\def\bra#1{\langle#1|} \def\ket#1{|#1\rangle}
\def\braket#1#2{\langle#1|#2\rangle}

\def\proj#1{\ket{#1}\!\bra{#1}}

\newtheorem{theo}{Theorem}

\newtheorem{remark}{Remark}
\newtheorem{defin}[theo]{Definition}

\newtheorem{lemma}[theo]{Lemma}

\def\id{{\mathbb I}}
\def\tr{\mbox{tr}}

\begin{document}
\title{Almost quantum correlations}
\author{Miguel Navascu\'es$^{1,2}$, Yelena Guryanova$^1$, Matty J. Hoban$^3$ and Antonio Ac\'in$^{3,4}$}
\affil{\small $^1$H.H. Wills Physics Laboratory, University of Bristol, Tyndall Avenue,\\ Bristol, BS8 1TL, United Kingdom\\$^2$Universitat Aut\`onoma de Barcelona, 08193 Bellaterra (Barcelona), Spain\\$^3$ICFO-Institut de Ciencies Fotoniques, Av. Carl Friedrich Gauss 3,\\ E-08860 Castelldefels, Barcelona, Spain\\$^4$ICREA--Institucio Catalana de Recerca i Estudis Avan\c{c}ats, E--08010 Barcelona, Spain \normalsize}
\date{}
\maketitle

\abstract{There have been a number of attempts to derive the set of quantum non-local correlations from reasonable physical principles. Here we introduce $\tilde{Q}$, a set of multipartite supra-quantum correlations that has appeared under different names in fields as diverse as graph theory, quantum gravity and quantum information science. We argue that $\tilde{Q}$ may correspond to the set of correlations of a reasonable physical theory, in which case the research program to reconstruct quantum theory from device-independent principles is met with strong obstacles. In support of this conjecture, we prove that $\tilde{Q}$ is closed under classical operations and satisfies the physical principles of Non-Trivial Communication Complexity, No Advantage for Nonlocal Computation, Macroscopic Locality and Local Orthogonality. We also review numerical evidence that almost quantum correlations satisfy Information Causality.}

\section{Introduction}

The validity of quantum mechanics in the microscopic and mesoscopic realm has been established up to incredible precision. However, despite the successes of the standard model of particle physics, gravity still does not quite fit into the picture. This fact, together with the complete absence of physical intuition in the historical formulation of quantum mechanics, has motivated a number of works where the Hilbert space structure of quantum theory was derived from first principles \cite{hardy,axioms, lluis, chiribella, hardy2}. This was done, not only with the intention to legitimize quantum theory, but also with the hope that some suitable relaxation of such principles would lead to interesting generalizations of quantum physics \cite{hardy}. 

Parallel to these efforts, there have been several attempts to recover the limits of quantum nonlocality from physical principles which can be formulated in a black-box scenario, with no reference to unobservable elements of the structure of the underlying physical theory. Axioms such as the Non-signalling Principle \cite{popescu}, Non-trivial Communication Complexity \cite{brassard}, No Advantage for Nonlocal Computation \cite{linden}, Information Causality \cite{marcin}, Macroscopic Locality \cite{mac_loc} and Local Orthogonality \cite{loc_orth} have been proposed to hold in all reasonable physical theories, and their associated constraints on the set of accessible nonlocal correlations have been studied thoroughly. To this day, however, it is an open question whether all these principles, or a subset of them, suffice to derive the set of quantum correlations.

In this paper we present an outer approximation to the set of quantum correlations which we term `almost quantum'. This approximation has appeared before in the scientific literature under different names and in completely different contexts, such as quantum information science \cite{NPA}, graph theory (see \cite{graph1,graph2} for references) and quantum gravity \cite{quant_grav}. Inspired by these surprising connections, we will present a series of results that support the conjecture that the set of almost quantum correlations does actually emerge from a reasonable physical theory. Firstly, we will prove that separate parties sharing a number of almost quantum boxes cannot, via post-selections and wirings, build new boxes outside the set: this hints that the almost quantum set actually corresponds to the set of nonlocal correlations of a specific model of reality. Secondly, we will argue that such a model, if it exists, must be physically compelling, since: a) we have numerical evidence that almost quantum correlations satisfy the principle of Information Causality; and b) we can prove that almost quantum correlations satisfy the rest of the physical principles stated above. Our purpose is two-fold: on one hand, we want to motivate the study of the almost quantum set in the hope that it inspires physical theories alternative to quantum mechanics. On the other hand, we want to argue that the program to recover quantum nonlocality via `reasonable' information-theoretic principles is fundamentally restricted, because our physical intuition about quantum nonlocality seems not to go beyond the almost quantum approximation.

It is worth noting at the outset that all principles except Local Orthogonality are defined in the two-party setting. It was shown in Ref. \cite{rodrigo} that truly multipartite principles are needed to retrieve quantum nonlocality. Here in our work, we treat principles defined in the bipartite setting \textit{solely} in the bipartite setting thus the argument of Ref. \cite{rodrigo} does not apply. On the other hand, in our work the one principle defined the multipartite setting, Local Orthogonality, is shown to be unable to retrieve quantum correlations. Our results in that case again do not rely on previous methods \cite{rodrigo}.

This paper is organized as follows: first, we will define the set of almost quantum correlations, provide a semidefinite programming characterization \cite{sdp} and comment on alternative definitions of it appearing in past literature. We will then show, in Section \ref{noquantum}, that the almost quantum set contains the quantum set strictly: indeed, even in the simplest nonlocality scenario, one can already find almost quantum distributions which cannot be approximated by quantum mechanical systems. In section \ref{closed_wirings}, we will prove that the set of almost quantum boxes is closed under classical operations, and, as such, may correspond to the nonlocal limits of a consistent physical theory. In sections \ref{ntcc}, \ref{ml}, \ref{nanc}, \ref{lo} we will rely on this result to prove that the set satisfies the principles of Non-trivial Communication Complexity \cite{brassard}, Macroscopic Locality \cite{mac_loc}, No Advantage for Nonlocal Computation \cite{linden} and Local Orthogonality \cite{loc_orth}. In Section \ref{ic} we will also discuss numerical evidence that suggests that almost quantum correlations also satisfy Information Causality \cite{marcin}. Finally, we will present our conclusions.

\section{$\tilde{Q}$, the almost quantum set: definition and SDP characterization}
\label{sec_defin}

Consider a scenario where $n$ parties conduct measurements $\bar{x}=(x_1,...,x_n)$ on their respective subsystems, obtaining outcomes $\bar{a}=(a_1,...,a_n)\in \{0,...,d-1\}^n$. Given $m\leq n$ parties, the pair $(\bar{a}|\bar{x})$, with the components of $\bar{a},\bar{x}$ labeled by each of the $m$ parties will be called an \emph{event}. For example, let $n=3$. Then, the event $(a_1,a_3|x_1,x_3)$ represents the physical situation in which parties 1 and 3 have measured $x_1,x_3$, and obtained, respectively, the outcomes $a_1,a_3$. Following \cite{loc_orth}, we say that two events $e\equiv (\bar{a}|\bar{x})$, $e'\equiv (\bar{a}'|\bar{x}')$ are \emph{locally orthogonal} (represented $e\perp e'$) if there is a common party $k$ such that $x_k=x_k'$, and $a_k\not=a_k'$.

\begin{defin}
\label{Q1AB}
Let $P(a_1,...,a_n|x_1,...,x_n)$ be an $n$-partite non-signalling distribution. We say that $P(a_1,...,a_n|x_1,...,x_n)$ is \emph{almost quantum} iff there exist a Hilbert space $\H$, a normalized state $\ket{\phi}\in \H$ and projector operators $\{E^{a,x}_k\}\subset B(\H)$ with the properties:

\begin{enumerate}[(i)]
\item 
\label{complete}
$\sum_aE_k^{a,x}=\id$, for all $x$, $k$.

\item
\label{permut} 
$E^{a_1,x_1}_1...E^{a_n,x_n}_n\ket{\phi}=E^{a_{\pi(1)},x_{\pi(1)}}_{\pi(1)}...E^{a_{\pi(n)},x_{\pi(n)}}_{\pi(n)}\ket{\phi}$, where $\pi\in S_n$ is an arbitrary permutation of the parties $\{1,...,n\}$.

\item
\label{prob}
$P(a_1,...,a_n|x_1,...,x_n)=\bra{\phi}\prod_{k=1}^nE^{a_k,x_k}_k\ket{\phi}$.

\end{enumerate}

\end{defin}

Any set of projectors $\{E^{a,x}_k\}$ and quantum state $\ket{\phi}$ satisfying the above conditions will be called an \emph{almost quantum representation} for $P(a_1,...,a_n|x_1,...,x_n)$. The set of all almost quantum distributions will be denoted by $\tilde{Q}$.

This definition must be contrasted with that of the set $Q$ of quantum correlations, namely:

\begin{defin}
\label{quantum}
Let $P(a_1,...,a_n|x_1,...,x_n)$ be an $n$-partite non-signalling distribution. We say that $P(a_1,...,a_n|x_1,...,x_n)$ is a quantum distribution iff there exist Hilbert spaces $\{\H_k\}_{k=1}^n$, a normalized state $\ket{\phi}\in \bigotimes_{k=1}^n\H_k$ and projector operators $\{E^{a,x}_k\}\subset B(\H_k)$ with the properties:

\begin{enumerate}[(i)]
\item 
$\sum_aE_k^{a,x}=\id_k$, for all $x$, $k$.

\item
$P(a_1,...,a_n|x_1,...,x_n)=\bra{\phi}\bigotimes_{k=1}^n E^{a_k,x_k}_k\ket{\phi}$.

\end{enumerate}

\end{defin}

Given a distribution $P(a_1,...,a_n|x_1,...,x_n)$ satisfying the conditions of Definition \ref{quantum}, it is immediate that $\{\tilde{E}^{a_k,x_k}_k\}$, $\ket{\phi}$, with $\tilde{E}^{a_k,x_k}_k\equiv\id^{\otimes k-1}\otimes E^{a_k,x_k}_k\otimes \id^{\otimes n-k}$, constitutes an almost quantum representation for $P(a_1,...,a_n|x_1,...,x_n)$. In other words, all quantum distributions are almost quantum, or $Q\subset \tilde{Q}$. In the next section we will see that this inclusion relation is strict.

The following lemma provides a semidefinite programming characterization of the set $\tilde{Q}$ of all almost quantum distributions.

\begin{lemma}
\label{defin_alter}
Let $P(a_1,...,a_n|x_1,...,x_n)$ be a non-signalling $n$-partite distribution. $P(a_1,...,a_n|x_1,...,x_n) \in \tilde{Q}$ iff, for any event $(\bar{a}|\bar{x})$ with $a_k\not=0$ for all parties $k$ involved, there exists a vector $\ket{\bar{a},\bar{x}}\in \H$ with the properties

\begin{enumerate}[(a)]
\item
\label{local_orth}
$\braket{\bar{a}',\bar{x}'}{\bar{a},\bar{x}}=0$, if $(\bar{a}'|\bar{x}')\perp (\bar{a}|\bar{x})$.
\item \label{probs} 
$P(\bar{a}|\bar{x})=\braket{\phi}{\bar{a},\bar{x}}$, where $\ket{\phi}$ is the (normalized) vector corresponding to the null event, i.e., none of the parties measures.

\item
\label{consist}
$\braket{\bar{a},\bar{a}',\bar{x},\bar{x}'}{\bar{a},\bar{a}'',\bar{x},\bar{x}''}=\braket{\bar{a}',\bar{x}'}{\bar{a},\bar{a}'',\bar{x},\bar{x}''}=\braket{\bar{a},\bar{a}',\bar{x},\bar{x}'}{\bar{a}'',\bar{x}''}$, where $(\bar{a}|\bar{x})$ is any event not involving the measuring parties in the events $(\bar{a}'|\bar{x}')$ and $(\bar{a}''|\bar{x}'')$.

\end{enumerate}

\end{lemma}

Note that any set of complex vectors subject to restrictions (\ref{local_orth}), (\ref{probs}), (\ref{consist}) implies the existence of real vectors subject to the same constraints. It follows that, in the above semidefinite program, all free variables can be taken real.

\begin{proof}
The right implication follows by defining $\ket{\bar{a},\bar{x}}\equiv \prod_kE^{a_k,x_k}\ket{\phi}$. 

Let us go for the left implication: consider the subspaces $V^{a,x}_k\equiv \mbox{span}\{\ket{a_k,\bar{a}',x_k,\bar{x}'}\}$. From condition (\ref{local_orth}), we have that $V^{a,x}_k\perp V^{a',x}_k$, for $a\not=a'$. It follows that the projectors $\tilde{E}^{a,x}_k\equiv \mbox{proj}(V^{a,x}_k)$ satisfy

\be
\tilde{E}^{a,x}_k\tilde{E}^{a',x}_k=\tilde{E}^{a,x}_k\delta_{a,a'}.
\ee

Now, the action of $\tilde{E}^{a_k,x_k}_k$ over the vector $\ket{\phi}$ is given by:

\be
\tilde{E}^{a_k,x_k}_k\ket{\phi}=\tilde{E}^{a_k,x_k}_k\ket{a_k,x_k}+\tilde{E}^{a_k,x_k}_k(\ket{\phi}-\ket{a_k,x_k})=\ket{a_k,x_k},
\ee

\noindent where we have used that $\ket{a_k,x_k}\in V^{a_k,x_k}_k$ in order to conclude $\tilde{E}^{a_k,x_k}_k\ket{a_k,x_k}=\ket{a_k,x_k}$. The second term $\tilde{E}^{a_k,x_k}_k(\ket{\phi}-\ket{a_k,x_k})$ vanishes, since

\be
\braket{a_k,\bar{a}',x_k,\bar{x}'}{\phi}=\braket{a_k,\bar{a}',x_k,\bar{x}'}{a_k,x_k}
\ee

\noindent by virtue of relation (\ref{consist}). Note that, also by condition (\ref{consist}), the last equality holds when we replace $\ket{\phi}$, $\ket{a_k,x_k}$ by $\ket{\bar{a},\bar{x}}$, $\ket{a_k,\bar{a},x_k,\bar{x}}$, where $(\bar{a}|\bar{x})$ is any event where party $k$ does not intervene. It follows that $\tilde{E}^{a_k,x_k}_k\ket{\bar{a},\bar{x}}=\ket{a_k,\bar{a},x_k,\bar{x}}$. By induction, we thus arrive at

\be
\prod_{k}\tilde{E}^{a_k,x_k}_k\ket{\phi}=\ket{\bar{a},\bar{x}},
\label{funda}
\ee

\noindent for $a_k\not=0$, where the product is taken in whatever order. Finally, define

\be
\tilde{E}^{0,x_k}_k\equiv \id-\sum_{a\not=0}\tilde{E}^{a,x_k}_k.
\ee

From eq. (\ref{funda}) and relation (\ref{probs}), it thus follows that the state $\ket{\phi}$ and the operators $\{\tilde{E}^{a,x}_k\}$ satisfy the conditions of definition \ref{Q1AB}.

\end{proof}

\begin{remark}
\label{certificate}
By the Gram decomposition \cite{horn}, the existence of a set of vectors satisfying the conditions of Lemma \ref{defin_alter} is equivalent to the existence of a positive semidefinite matrix $\Gamma$, with rows and columns labeled by events $(\bar{a}|\bar{x})$ with $a_k\not=0$ for all parties involved, and such that

\begin{enumerate}[(a)]
\item
$\Gamma_{(\bar{a}'|\bar{x}'),(\bar{a}|\bar{x})}=0$, if $(\bar{a}'|\bar{x}')\perp (\bar{a}|\bar{x})$.
\item
$\Gamma_{\phi,\phi}=1$.
\item 
$P(\bar{a}|\bar{x})=\Gamma_{\phi,(\bar{a}|\bar{x})}$.

\item
$\Gamma_{(\bar{a},\bar{a}'|\bar{x},\bar{x}'),(\bar{a},\bar{a}''|\bar{x},\bar{x}'')}=\Gamma_{(\bar{a}'|\bar{x}'),(\bar{a},\bar{a}''|\bar{x},\bar{x}'')}=\Gamma_{(\bar{a},\bar{a}'|\bar{x},\bar{x}'),(\bar{a}''|\bar{x}'')}$, where $(\bar{a}|\bar{x})$ is any event not involving the measuring parties in the events $(\bar{a}'|\bar{x}')$ and $(\bar{a}''|\bar{x}'')$.

\end{enumerate}

\noindent We will call any such matrix an \emph{almost quantum certificate} for $P(\bar{a}|\bar{x})$.
\end{remark}

In the bipartite case, Lemma \ref{defin_alter} allows us to identify $\tilde{Q}$ with the set $Q^{1+AB}$, defined in \cite{NPA2} as an approximation to the set of quantum correlations for applications in quantum information theory. $Q^{1+AB}$ can also be interpreted, by Definition \ref{Q1AB}, as the set of bipartite probability distributions admitting a strongly positive decoherence functional \cite{quant_grav}, see \cite{joe} for a proof\footnote{The decoherence functional approach is a sum-of-histories-based relaxation of quantum theory introduced in \cite{hartle1} as an attempt to handle fluctuating space-time geometries.}. Finally, in any non-locality scenario defined by the hypergraph $H$, the set $\tilde{Q}$ can be identified with the set of probabilistic models $p$ such that $\vartheta(\mbox{Ort}(H),p)=1$ \cite{graph1,graph2}. Here the $\vartheta(\bullet,p)$ denotes the $p$-weighted Lov\'asz number of the graph $\bullet$; and $\mbox{Ort}(H)$, the non-orthogonality graph of $H$, see \cite{graph2} for the corresponding definitions. As we can see, the almost quantum set has been independently derived in a variety of contexts. It is hence not unreasonable to presume that there is something `natural' about this set.

\section{$\tilde{Q}$ is supra-quantum}
\label{noquantum}

There already exists in the literature numerical evidence that $\tilde{Q}$ is supra-quantum \cite{NPA2,pal} but now we give an analytical proof. To prove that $\tilde{Q}\not =Q$, it is enough to consider a nonlocality scenario with two inputs $a_k \in \{0,1\}$, two outputs $x_k \in \{0,1\}$ and two parties ($k=1, 2$). Due to normalization and no-signalling constraints, any probability distribution in this scenario can be written as an 8-dimensional vector

\begin{align}
\overline{p}\equiv \left( P_1(1 |0) , P_1(1 |1) ,P_2(1 |0) ,P_2(1 |1) , P(1,1|0,0), P(1,1|1,0), P(1,1|0,1), P(1,1|1,1) \right),
\end{align}

\noindent where $P_k$ denotes the marginal probability distribution of party $k$.

Consider the Bell inequality $B(\overline{p}) \equiv \overline{b}\cdot \overline{p}$, with
\begin{align}\label{bell}
\overline{b} =  \left(-\frac{30}{31},\frac{167}{9},\frac{167}{9},-\frac{30}{31},-\frac{174}{11},-\frac{244}{23},\frac{74}{11},-\frac{174}{11}\right).
\end{align}

\noindent To estimate its minimal quantum value, we use the fact that all extreme distributions in the two inputs/two outputs scenario can be obtained by measuring a two-qubit system with the following projectors \cite{extreme_2222}:
\begin{eqnarray}
&\tilde{E}_1^{a=1,x=0} = \proj{1}\otimes \id_2, \tilde{E}_1^{a=1,x=1} = \proj{\psi_1}\otimes \id_2,\nonumber\\
&\tilde{E}_2^{a=1,x=0} = \id_2\otimes \proj{1}, \tilde{E}_2^{a=1,x=1} = \id_2\otimes\ket{\psi_2}\bra{\psi_2},
\label{general_meas}
\end{eqnarray}
\noindent where $\ket{\psi_{1,2}} = \cos(\theta_{1,2})\ket{0} + \sin(\theta_{1, 2})\ket{1}$. 

For simplicity, denote $\tilde{E}_k^{a=1,x}$ as $\tilde{E}_k^{x}$ and define the Bell operator 

\be
M(\theta_1,\theta_2)\equiv \overline{b}\cdot \left(E_1^{0},E_1^{1},E_2^{0},E_2^{1}, E_1^0 E_2^0, E_1^1 E_2^0, E_1^0 E_2^1, E_1^1 E_2^1, \right).
\ee

\noindent Using the determinant criterion \cite{horn}, it can be verified that

\be
M(\theta_1,\theta_2)+\id_4>0,
\ee

\noindent for all $\theta_1,\theta_2\in [0,2\pi)$; it follows that $B(\overline{p})>-1$ for all $\overline{p}\in Q$.

Now, consider the distribution
\begin{align}
\overline{p}_{\tilde{Q}} = \left(
\begin{array}{ccccccccc}
  \frac{9}{20} & \frac{2}{11} & \frac{2}{11} & \frac{9}{20} & \frac{22}{125} & \frac{\sqrt{2}}{9} & \frac{37}{700} & \frac{22}{125} \\
\end{array}
\right). 
\end{align}
\noindent This distribution lives in $\tilde{Q}$, since it admits the almost quantum certificate

\begin{align}
\Gamma=\left(
\begin{array}{ccccccccc}
 1 & \frac{9}{20} & \frac{2}{11} & \frac{2}{11} & \frac{9}{20} & \frac{22}{125} & \frac{\sqrt{2}}{9} & \frac{37}{700} & \frac{22}{125} \\
 \frac{9}{20} & \frac{9}{20} & \frac{17}{155} & \frac{22}{125} & \frac{37}{700} & \frac{22}{125} & \frac{\sqrt{33}}{40} & \frac{37}{700} & \frac{\sqrt{71}}{100} \\
 \frac{2}{11} & \frac{17}{155} & \frac{2}{11} & \frac{\sqrt{2}}{9} & \frac{22}{125} & \frac{\sqrt{33}}{40} & \frac{\sqrt{2}}{9} & \frac{\sqrt{71}}{100} & \frac{22}{125} \\
 \frac{2}{11} & \frac{22}{125} & \frac{\sqrt{2}}{9} & \frac{2}{11} & \frac{17}{155} & \frac{22}{125} & \frac{\sqrt{2}}{9} & \frac{\sqrt{71}}{100} & \frac{\sqrt{33}}{40} \\
 \frac{9}{20} & \frac{37}{700} & \frac{22}{125} & \frac{17}{155} & \frac{9}{20} & \frac{\sqrt{71}}{100} & \frac{\sqrt{33}}{40} & \frac{37}{700} & \frac{22}{125} \\
 \frac{22}{125} & \frac{22}{125} & \frac{\sqrt{33}}{40} & \frac{22}{125} & \frac{\sqrt{71}}{100} & \frac{22}{125} & \frac{\sqrt{33}}{40} & \frac{\sqrt{71}}{100} & \frac{21}{158} \\
 \frac{\sqrt{2}}{9} & \frac{\sqrt{33}}{40} & \frac{\sqrt{2}}{9} & \frac{\sqrt{2}}{9} & \frac{\sqrt{33}}{40} & \frac{\sqrt{33}}{40} & \frac{\sqrt{2}}{9} & \frac{4}{53} & \frac{\sqrt{33}}{40} \\
 \frac{37}{700} & \frac{37}{700} & \frac{\sqrt{71}}{100} & \frac{\sqrt{71}}{100} & \frac{37}{700} & \frac{\sqrt{71}}{100} & \frac{4}{53} & \frac{37}{700} & \frac{\sqrt{71}}{100} \\
 \frac{22}{125} & \frac{\sqrt{71}}{100} & \frac{22}{125} & \frac{\sqrt{33}}{40} & \frac{22}{125} & \frac{21}{158} & \frac{\sqrt{33}}{40} & \frac{\sqrt{71}}{100} & \frac{22}{125} \\
\end{array}
\right).
\end{align}

\noindent However, one can check that
\begin{align}
B(\overline{p}_{\tilde{Q} }) \approx -1.052.
\end{align}

\noindent This value is smaller than the quantum minimum, and thus $\overline{p}_{\tilde{Q} }\not\in Q$.

\section{$\tilde{Q}$ is closed under classical operations}
\label{closed_wirings}

Two parties sharing a number of independent non-local boxes can `wire' them together using classical circuitry to generate a new effective bipartite box \cite{closed}. Clearly, the set of correlations of any physical theory must be closed under these operations (i.e., closed under wirings). As shown in \cite{closed}, closure under wirings is a highly non-trivial property and fairly natural polytopes in the two inputs/two outputs bipartite Bell scenario fail to satisfy it.

In the same line, we next prove a result that suggests that $\tilde{Q}$ does correspond to the set of non-local correlations of a (yet unknown) physical theory: namely, we show that the manipulation of whatever number of almost quantum boxes by $n$ parties cannot generate effective non-local boxes outside the set $\tilde{Q}$.

We will divide the proof in three parts: first, we will prove that $\tilde{Q}$ is closed under post-selection; then, that it is closed under composition and finally, that it is closed under grouping of the parties. Since these three operations exhaust the set of actions which we can perform on a collection of boxes in $\tilde{Q}$, it follows that $\tilde{Q}$ can be regarded as a closed theory.

\begin{lemma}{Closure under post-selection}\\
Let $P(a_1,...,a_n|x_1,...,x_n)$ be almost quantum. Then, the post-selection $P(a_2,...,a_n|x_1,...,x_n,a_1)$ is almost quantum.

\end{lemma}

\begin{proof}
If $P(a_1,...,a_n|x_1,...,x_n)$ is almost quantum, there exist a set of operators $\{E^{a,x}_k\}$ and a quantum state $\ket{\phi}$ with the properties (\ref{complete}), (\ref{permut}), (\ref{prob}). Define $\ket{\phi'}\equiv E^{a_1,x_1}_k\ket{\phi}/\|E^{a_1,x_1}_1\ket{\phi}\|_2$, noting that $\|E^{a_1,x_1}_1\ket{\phi}\|_2=\sqrt{P(a_1|x_1)}$. Then it is straightforward that 

\be
E^{a_2,x_2}_2...E^{a_n,x_n}_n\ket{\phi'}=E^{a_{\pi(2)},x_{\pi(2)}}_{\pi(2)}...E^{a_{\pi(n)},x_{\pi(n)}}_{\pi(n)}\ket{\phi'}
\ee

\noindent for any permutation $\pi\in S_n$ with $\pi(1)=1$. Also, 

\be
\bra{\phi'}\prod_{k=2}^nE^{a_k,x_k}_k\ket{\phi'}=\frac{P(a_1,...,a_n|x_1,...,x_n)}{P(a_1|x_1)}=P(a_2,...,a_n|x_1,...,x_n,a_1).
\ee

\noindent The last two conditions imply that $P(a_2,...,a_n|x_1,...,x_n,a_1)$ is almost quantum.

\end{proof}

\begin{lemma}{Closure under composition}\\
Let $P_A,P_B\in \tilde{Q}$ be $n_A$-partite and $n_B$-partite independent boxes. Then, the $n_A+n_B$-partite box $P_A(\bar{a}|\bar{x}) P_B(\bar{b}|\bar{y})$ that results from the composition is almost quantum.

\end{lemma}

\begin{proof}
Let $\{E^{a,x}_k\}, \ket{\phi_A}$ ($\{F^{(b|y)}_k\}, \ket{\phi_B}$) be an almost quantum representation for the distribution $P_A$ ($P_B$). Define the state $\ket{\phi}=\ket{\phi_A}\otimes \ket{\phi_B}$ and the projectors

\be
\tilde{E}^{a|x}_k=E^{a,x}_k\otimes \id,\tilde{F}^{(b|y)}_k=\id \otimes F^{(b|y)}_k.
\ee

\noindent It is straightforward to verify that such state and projectors are an almost quantum representation for the distribution $P_A(\bar{a}|\bar{x}) P_B(\bar{b}|\bar{y})$.

\end{proof}

\begin{lemma}{Closure under grouping of parties}\\
Let $P(a_1,...,a_m,b_1,...,b_n|x_1,...,x_m,y_1,...,y_m)\in \tilde{Q}$. Suppose that the first $m$ parties join and apply wirings to generate the new distribution $P(\tilde{a},b_1,...,b_n|\tilde{x},y_1,...,y_n)$, where $\tilde{x}, \tilde{a}$ denote, respectively, a wiring and an outcome. Then, $P(\tilde{a},b_1,...,b_n|\tilde{x},y_1,...,y_n)\in \tilde{Q}$.

\end{lemma}

\begin{proof}
Let $\{E^{a_k,x_k}_k\}_{k=1}^m, \{F^{b_k,y_k}_k\}_{k=1}^n$, $\ket{\phi}$ be an almost quantum representation for $P$. Now, consider the \emph{fine-grained} wirings effected by the first $m$ parties, i.e., those wirings where any two different sequences of outcomes correspond to different `effective' or `final' outcomes. Like in quantum mechanics, the outcome $\tilde{a}$ of a wiring $\tilde{x}$ can be represented by the product of projectors $E^{\tilde{a},\tilde{x}}\equiv\prod_{k}E^{a_k,x_k}_k$ such that

\be
P(\tilde{a},b_1,...,b_n|\tilde{x},y_1,...,y_n)=\bra{\phi}\prod_{j=1}^nF^{b_j,y_j}_jE^{\tilde{a},\tilde{x}}\ket{\phi}.
\ee

\noindent Now, for any event $(\tilde{a},\bar{b}|\tilde{x},\bar{y})$, define the vector

\be
\ket{\tilde{a},\bar{b},\tilde{x},\bar{y}}\equiv E^{\tilde{a},\tilde{x}}\cdot \prod_{j}F^{b_j,y_j}_j\ket{\phi}.
\ee

These vectors obviously satisfy conditions (\ref{probs}), (\ref{consist}) of Lemma \ref{defin_alter}. To see that they also satisfy condition (\ref{local_orth}), note that, for any two different outcomes $\tilde{a}, \tilde{a}'$ of a fine-grained wiring there exists a party $k\in \{1,...,m\}$ who performed the same measurement $x_k$ and obtained different outcomes $a_k,a'_k$. It follows that 

\be
\braket{\tilde{a},\bar{b},\tilde{x},\bar{y}}{\tilde{a}',\bar{b}',\tilde{x},\bar{y}'}=\bra{\phi}E^{a_k,x_k}_k(E) (F) (E') (F') E^{a'_k,x_k}_k\ket{\phi}=0.
\ee

\noindent By Lemma \ref{defin_alter}, the new distribution is hence almost quantum. 

That the resulting distribution remains almost quantum when the wirings are not fine-grained, i.e., when many different measurement paths are associated to the same outcome, follows from the fact that $\tilde{Q}$ is closed under grouping of various outcomes. From Definition \ref{Q1AB} this is almost immediate: let $a\not=a'$, and let $E^{a,x_k}_k$, $E^{a',x_k}_k$ be the projectors associated to the events $(a|x_k)$, $(a'|x_k)$ in the almost quantum representation of the considered distribution $P$. Then, the operator $E^{a,x_k}_k+E^{a',x_k}_k$ is a projector, orthogonal to the projectors $\{E^{a_k,x_k}_k:a_k\not=a,a'\}$ and satisfying conditions (\ref{permut}), (\ref{prob}) of definition \ref{Q1AB}. Finally, it is obvious that the projectors $\{E^{a,x_k}_k+E^{a',x_k}_k\}\cup \{E^{a_k,x_k}_k:a_k\not=a,a'\}$ satisfy (\ref{complete}).

\end{proof}

\section{$\tilde{Q}$ satisfies non-trivial communication complexity}
\label{ntcc}

Roughly speaking, the axiom of Non-Trivial Communication Complexity (NTCC) states that two spatially separated parties, call them Alice and Bob, cannot compute arbitrary boolean functions with fixed probability greater than 1/2 for all input sizes \cite{brassard}. More concretely, suppose that Alice and Bob are respectively distributed the strings of bits $\bar{x},\bar{y}\in\{0,1\}^n$. Bob's task is to compute the function $f(\bar{x},\bar{y})\in \{0,1\}$, and, to that effect, Alice is allowed to transmit him \emph{one} bit of information. For a particular protocol, call $p(\bar{x},\bar{y})$ the probability that Bob succeeds when the inputs are $\bar{x},\bar{y}$. NTCC then implies that there exists a family of functions $\{f_n:\{0,1\}^n\times\{0,1\}^n\to \{0,1\}\}$ such that no communication protocol can succeed with probability $p(\bar{x},\bar{y})>p>1/2$ \emph{independent of $n$} for all $\bar{x},\bar{y}\in \{0,1\}^n$ and all input sizes $n$. In \cite{brassard} it is shown that boxes with a Clauser-Horne-Shimony-Holt (CHSH) parameter \cite{chsh} greater than $4\sqrt{2/3}$ could be used to devise protocols which violate this principle. NTCC thus imposes non-trivial constraints on the set of non-local correlations.

In the next lines, we will prove that $\tilde{Q}$ satisfies NTCC. We will do so by showing that two parties sharing a number of bipartite boxes $\{Q_i(a,b|x,y)\}_i\subset \tilde{Q}$ cannot solve the inner product problem with worst-case probability $p$ using a fixed amount of 1-way communication (not necessarily just one bit). This is a consequence of the following theorem:

\begin{theo}
Let Alice and Bob share a number of bipartite boxes $\{Q_i(a,b|x,y)\}_i\subset \tilde{Q}$, and consider 1-way communication protocols where they are distributed the $n$-bit strings $\bar{x}, \bar{y}$; Alice is allowed to transmit Bob $m$ bits; and Bob must output a guess $b$ of the inner product $\bar{x}\cdot \bar{y} \mbox{ (mod 2)}$. Then, the worst-case probability of success $p$ must satisfy:

\be
\frac{1}{2^{n-m}}\geq (2p-1)^2.
\label{inequa}
\ee
\end{theo}

\noindent If $\tilde{Q}$ had fixed 1-way communication complexity, i.e., if the same number $m$ of bits sufficed for any input size $n$, then, for any $p>1/2$, the above inequality would be violated by taking $n$ high enough.

Following the lines of \cite{nielsen}, we will prove the theorem by showing that a protocol allowing Alice and Bob to compute the inner product with great probability implies the existence of a non-signalling box that allows Bob to guess Alice's $n$-bit input $\bar{x}$ from her $m$-bit message. We will then prove that this is impossible for $m<n$. The proof hence relies on the next two lemmas:

\begin{lemma}
\label{quasi_nielsen}
Let Alice and Bob share a distribution $P(\bar{a},b|\bar{x},[\bar{z}, \bar{y}])\in \tilde{Q}$, where $\bar{x}$ and $[\bar{z},\bar{y}]$ denote, respectively, Alice's and Bob's measurement settings, and $\bar{a},\bar{z}\in \{0,1\}^m$, $\bar{x},\bar{y}\in \{0,1\}^n$, $b\in\{0,1\}$. Then, there exists a non-signalling bipartite distribution $P'(\bar{a},\bar{c}|\bar{x},\bar{z})$, with $\bar{c}\in \{0,1\}^n$, such that

\be
\sum_{\bar{a}}P'(\bar{a},\bar{x}|\bar{x},\bar{a})\geq \sum_{\bar{a}}P(\bar{a}|\bar{x})(2P(\mbox{success}|\bar{x},\bar{a})-1)^2,
\label{robustness}
\ee

\noindent with

\be
P(\mbox{success}|\bar{x},\bar{a})\equiv\frac{1}{2^n}\sum_{\bar{y}}P(b=\bar{x}\cdot\bar{y}|\bar{x},\bar{a},[\bar{a},\bar{y}]).
\ee

\end{lemma}

\noindent For a proof, see Appendix \ref{non_triv_ap}.

\begin{lemma}
\label{no_sig}
Let $P(\bar{a},\bar{c}|\bar{x},\bar{z})$ be a non-signalling probability distribution, with $\bar{a},\bar{z}\in \{0,1\}^m$ and $\bar{c},\bar{x}\in \{0,1\}^n$. Then,

\be
\frac{1}{2^n}\sum_{\bar{x}}\sum_{\bar{a}}P(\bar{a},\bar{x}|\bar{x},\bar{a})\leq \frac{1}{2^{n-m}}.
\ee

\end{lemma}

\begin{proof}
Consider the following protocol: Alice is distributed the random sequence of bits $\bar{x}$ with probability $\frac{1}{2^n}$, she measures $\bar{x}$ and obtains the result $\bar{a}\in \{0,1\}^m$ with probability $P(\bar{a}|\bar{x})$. Then Bob makes a completely random guess $\bar{a}'$ on the value of $\bar{a}$ and inputs it into his box. Bob's outcome, $\bar{c}$, will be Bob's guess on the value of $\bar{x}$. With probability $1/2^m$, Bob's guess on $\bar{a}$ will be correct, i.e., $\bar{a}=\bar{a}'$, in which case he will guess Alice's input with probability $P(\bar{a},\bar{x}|\bar{x},\bar{a})/P(\bar{a}|\bar{x})$. It follows that Bob's strategy to guess $\bar{x}$ will succeed with probability at least

\be
\frac{1}{2^n}\sum_{\bar{x}}\sum_{\bar{a}}P(\bar{a}|\bar{x})\frac{1}{2^m}\frac{P(\bar{a},\bar{x}|\bar{x},\bar{a})}{P(\bar{a}|\bar{x})}.
\ee

On the other hand, in this protocol no information has been transmitted to Bob. By no-signalling, it follows that, no matter what his strategy is, Bob will guess Alice's input $\bar{x}$ with probability $1/2^n$. This, together with the above lower bound on the success probability, implies the main claim.

\end{proof}

Let us put everything together: first, note that we can always map any inner product protocol to a non-locality scenario where Alice's measurement is labeled by $\bar{x}$, and her outcome, by the $m$-bit message $\bar{a}$ that she sends to Bob. Bob's measurement setting is labeled by his bit-string $\bar{y}$ together with Alice's message $\bar{z}$, i.e., Bob's measurement is described by the pair $[\bar{z},\bar{y}]$. Bob's outcome will be the bit $b$ that (he hopes!) satisfies $b=\bar{x}\cdot\bar{y}$. Since $\tilde{Q}$ is closed under wirings, Alice and Bob's box $P(\bar{a},b|\bar{x},[\bar{z}, \bar{y}])$ must necessarily belong to $\tilde{Q}$.

Now, suppose that Alice and Bob can compute the inner product probabilistically with $m$ bits of communication via the box $P(\bar{a},b|\bar{x},[\bar{z}, \bar{y}])\in \tilde{Q}$, with worst-case success probability $p>\frac{1}{2}$. That is,

\be
\sum_{\bar{a}}P(\bar{a},b=\bar{x}\cdot \bar{y}|\bar{x},[\bar{a},\bar{y}])\geq p,
\ee

\noindent for all $\bar{x}, \bar{y}$. Such a box therefore satisfies

\be
\sum_{\bar{a}}P(\bar{a}|\bar{x})(2P(\mbox{success}|\bar{x},\bar{a})-1)\geq 2p-1>0,
\ee

\noindent with $P(\mbox{success}|\bar{x},\bar{a})$ defined as in Lemma \ref{quasi_nielsen}. It follows that

\be
\sum_{\bar{a}}P(\bar{a}|\bar{x})(2P(\mbox{success}|\bar{x},\bar{a})-1)^2 \geq (2p-1)^2.
\ee

\noindent By Lemma \ref{quasi_nielsen} we thus have that there exists a non-signalling distribution $P'(\bar{a},\bar{c}|\bar{x},\bar{z})$ such that

\be
\sum_{\bar{a}}P'(\bar{a},\bar{x}|\bar{x},\bar{a})\geq (2p-1)^2.
\ee

\noindent Dividing by $\frac{1}{2^n}$, summing over $\bar{x}$ and applying Lemma \ref{no_sig}, we arrive at eq. (\ref{inequa}).

\section{$\tilde{Q}$ satisfies Macroscopic Locality}
\label{ml}

Macroscopic Locality (ML) states that coarse-grained extensive measurements of $N$ independent particle pairs must admit a local hidden variable model in the limit $N\to \infty$, see \cite{mac_loc} for details. ML is justified on the grounds that any reasonable physical theory must have a classical limit; ergo, `natural' macroscopic experiments should be describable by a classical theory, and consequently their associated statistics must be local realistic.

In \cite{mac_loc} it is shown that the set of bipartite distributions compatible with this principle corresponds to $Q^1$, a semidefinite programming relaxation of the set of quantum correlations firstly introduced in \cite{NPA}. $Q^1$ is defined as the set of all non-signalling correlations $P(a,b|x,y)$ such that there exists a positive semidefinite matrix $\gamma$ -with rows and columns labeled by the events $\{\phi\}\cup\{(a|x):a\not=0\}\cup\{(b|y):b\not=0\}$- of the form:

\be
\gamma=\left(\begin{array}{ccc}1&\vec{p}^T_A&\vec{p}^T_B\\ \vec{p}_A&Q&P\\\vec{p}_B&P^T&R\end{array}\right),
\label{comp_mat}
\ee

\noindent where $\vec{p}_A$ ($\vec{p}_B$) is Alice's (Bob's) vector of marginal probabilities, and

\be
Q_{(a,x),(a',x)}=P(a|x)\delta_{a,a'}, P_{(a,x),(b,y)}=P(a,b|x,y), R_{(b,y),(b',y)}=P(b|y)\delta_{b,b'}.
\label{linear_cons}
\ee

From Remark \ref{certificate}, it is easy to see that any bipartite distribution $P(a,b|x,y)\in \tilde{Q}$ satisfies ML. Indeed, let $\Gamma$ be an almost quantum certificate for $P(a,b|x,y)$. Then one can verify that the submatrix $\tilde{\gamma}$ of $\Gamma$ given by $\tilde{\gamma}=\{\Gamma_{\alpha,\beta}:\alpha,\beta\in \{\phi\}\cup\{(a|x):a\not=0\}\cup\{(b|y):b\not=0\}\}$ satisfies conditions (\ref{linear_cons}). Also, since it is a submatrix of $\Gamma$, it is positive semidefinite.

\section{$\tilde{Q}$ satisfies No Advantage for Nonlocal Computation}
\label{nanc}

Nonlocal computation is an information processing task introduced in \cite{linden} by Linden \emph{et al.}. In this primitive, an $n$-bit string $z\in\{0,1\}^{n}$ is distributed with prior probability $\tilde{p}(z)\geq 0$, with $\sum_{z}\tilde{p}(z)=1$. The goal behind nonlocal computation is to have two non-communicating parties, Alice and Bob, evaluate the Boolean function $f:\{0,1\}^{n}\rightarrow\{0,1\}$ on $z$ while learning nothing about the value of $z$. For that purpose, a fully random $n$-bit string $x$ is generated and sent to Alice, while Bob receives the bit string $y\equiv z\oplus x$. Given inputs $x,y$, Alice and Bob's task is to produce two binary outputs $a,b$ such that $a\oplus b=f(x\oplus y)=f(z)$.

The figure of merit of nonlocal computation is Alice and Bob's average success at computing $f$, given by the expression:
\begin{equation}
P(f)=\frac{1}{2^{n}}\sum_{xy}\tilde{p}(x\oplus y)P(a\oplus b=f(x \oplus y)|xy).
\end{equation}
If Alice and Bob are restricted to use classical resources, that is, if they can only have shared randomness between them, then the maximum probability of success is:
\begin{equation}\label{classical}
P_{C}(f)=\frac{1}{2}\left(1+\underset{u\in\{0,1\}^{n}}{\textrm{max}}\left|\sum_{z}(-1)^{f(z)+u\cdot z}\tilde{p}(z)\right|\right).
\end{equation}
Furthermore, as shown in \cite{linden}, if the two parties share entangled states, they cannot do any better. On the other hand, there exist non-signalling resources that would allow them to beat the value $P_{C}(f)$. 

Next we will show that, even when the two parties are distributed general ML distributions, they still cannot beat the classical value given by eq. (\ref{classical}). Since by the last section $\tilde{Q}\subset Q^1$, it follows that $\tilde{Q}$ also satisfies NANLC.

Let $P(a,b|x,y)\in Q^{1}$, with $a,b\in \{0,1\}$, and let $\gamma$ be a positive semidefinite matrix of the form (\ref{comp_mat}). From its Gram decomposition $\gamma_{\alpha,\beta}=\braket{\alpha}{\beta}$, we obtain the vectors $\ket{\phi}$, $\ket{1,x}$, $\ket{1,y}$. Defining $\ket{0,x}\equiv \ket{\phi}-\ket{1,x}$, $\ket{0,y}\equiv \ket{\phi}-\ket{1,y}$, we have that $\braket{a,x}{a',x}=P(a|x)\delta_{aa'}$, $\braket{b,y}{b',y}=P(b|y)\delta_{bb'}$ and $\braket{a,x}{b,y}=P(a,b|x,y)$.

In terms of these vectors, $P(a\oplus b=f(x \oplus y)|xy)$ can be written as
\begin{eqnarray}
P(a\oplus b=f(x \oplus y)|xy)&=&\sum_{a,b}\delta^{a\oplus b}_{f(x \oplus y)}\braket{a,x}{b,y}\\
&=&\frac{1}{2}\sum_{a,b}\left(1+(-1)^{a+b+f(x \oplus y)}\right)\braket{a,x}{b,y}\\
&=&\frac{1}{2}\left(1+\sum_{a,b}(-1)^{a+b+f(x \oplus y)}\braket{a,x}{b,y}\right).
\end{eqnarray}
The average success probability for the nonlocal computation of $f$ is hence
\begin{eqnarray}
P_{Q^{1}}(f)&=&\frac{1}{2^{n+1}}\sum_{x,y}\tilde{p}(x\oplus y)\left(1+\sum_{a,b}(-1)^{a+b+f(x\oplus y)}\braket{a,x}{b,y}\right)\\
&=&\frac{1}{2}+\frac{1}{2^{n+1}}\sum_{a,b,x,y}\tilde{p}(x\oplus y)(-1)^{a+b+f(x\oplus y)}\braket{a,x}{b,y}.
\end{eqnarray}
Following Ref. \cite{linden}, we now introduce the vectors $\ket{\alpha},\ket{\beta}$ and the operator $\Phi$:
\begin{eqnarray}
\langle\alpha|&=&\frac{1}{\sqrt{2^{n}}}\sum_{x}\left(\sum_{a}(-1)^{a}\bra{a,x}\right)\otimes\langle x|\\
|\beta\rangle&=&\frac{1}{\sqrt{2^{n}}}\sum_{y}\left(\sum_{b}(-1)^{b}\ket{b,y}\right)\otimes|y\rangle\\
\Phi &=&\sum_{x,y}(-1)^{f(x\oplus y)}\tilde{p}(x \oplus y)|x\rangle\langle y|,
\end{eqnarray}
where $x$ and $y$ label the computational basis states. Given these expressions, we observe that
\begin{equation}
P_{Q^{1}}(f)=\frac{1}{2}\left(1+\langle\alpha|(\mathbb{I}\otimes\Phi)|\beta\rangle\right).
\end{equation}
Since $\langle\alpha|$ and $|\beta\rangle$ are normalized vectors, we have that
\begin{equation}
P_{Q^{1}}(f)\leq\frac{1}{2}\left(1+|\langle\alpha||\|\mathbb{I}\otimes\Phi\|||\beta\rangle\right|)= \frac{1}{2}\left(1+\|\Phi\|\right).
\end{equation}
As shown in Ref. \cite{linden}, this last expression is upper-bounded by \eqref{classical}. Therefore, $Q^{1}$ does no better than classical physics in nonlocal computation.

\section{$\tilde{Q}$ satisfies Local Orthogonality}
\label{lo}
Consider any set $E$ of pairwise locally orthogonal events, as defined in Section \ref{sec_defin} (i.e., $e\perp e'$ iff $\exists k$ s.t. $x_k=x'_k$, $a_k\not=a_k'$). Local Orthogonality (LO) \cite{loc_orth} states that the sum of the probabilities of each event in $E$ cannot exceed 1, that is,

\be
\sum_{e\in E} P(e)\leq 1.
\ee

\noindent This principle is equivalent to demanding that distributed guessing problems which are maximally difficult classically remain so when the parties involved are assisted with non-local resources, see \cite{loc_orth} for an explanation.

The proof that $\tilde{Q}$ satisfies LO has already appeared in \cite{graph2}, but, for the sake of completeness, we present here an alternative derivation that does not rely on graph-theoretical concepts. Let $P(a_1,...,a_n|x_1,...,x_n)\in \tilde{Q}$, and let $\ket{\phi}$, $\{E^{a,x}_k\}$ be its almost quantum representation. For any event $(\bar{a}|\bar{x})$, define the vector $\ket{e}\equiv \prod_k E^{a_k,x_k}_k\ket{\phi}$. It is easy to see that, for any two locally orthogonal events, $e\perp e'$, $\braket{e}{e'}=0$. Moreover, one can prove that, for any event $e$, 

\be
\ket{\phi}=\ket{e}+\ket{e^\perp},
\ee
\noindent with $\braket{e}{e^\perp}=0$. This follows from the fact that

\be
\sum_{\bar{a}}\left(\prod_kE_k^{a_k,x_k}\ket{\phi}\right)=\prod_k\left(\sum_{a_k}E_k^{a_k,x_k}\right)\ket{\phi}=\ket{\phi},
\ee

\noindent and the observation that each of the vectors in the left hand side of the above equation is orthogonal to all the others\footnote{This follows from Property (\ref{permut}) in Definition \ref{Q1AB} and the relation $E_k^{a_k,x_k}E_k^{a'_k,x_k}=\delta_{a_k,a_k'}E_k^{a_k,x_k}$.}. For any event $e$, we thus have that $P(e)=\bra{\phi}\Pi_e\ket{\phi}$, with $\Pi_e\equiv\frac{\proj{e}}{\braket{e}{e}}$. 

Now, consider an arbitrary set $E$ of locally orthogonal events. By the above considerations we have that

\be
\sum_{e\in E}P(e)=\sum_{e\in E}\bra{\phi}\Pi_e\ket{\phi}=\bra{\phi}\sum_{e\in E} \Pi_e\ket{\phi}\leq 1,
\ee

\noindent where the last inequality is due to the fact that the norm of a sum of orthogonal projectors is either 0 or 1.

\section{Evidence that $\tilde{Q}$ satisfies Information Causality}
\label{ic}

Consider a bipartite scenario, similar to that of Section \ref{ntcc}, where Alice (Bob) receives a completely random $n$-bit string (a random number) $x_1,...,x_n$ ($k\in \{1,...,n\}$).  Bob's task consists in making a guess $b$ for Alice's bit $x_k$. To aid him, Alice is allowed to send Bob $m<n$ bits of information. 

Note that, if this protocol could be played perfectly, as soon as Alice sent her $m$ bits, Bob would be in possession of a box that would \emph{potentially} contain $n$ of Alice's bits. In this scenario, however, one would intuitively expect Bob's system to hold no more than $m$ potential bits of information. The principle of Information Causality \cite{marcin} tries to capture this intuition by stating that:

\be
\sum_{k=1}^nI(b:x_k|k)\leq m.
\ee

\noindent Here $I(A:B)$ denotes the mutual information between the random variables $A$ and $B$, i.e., $I(A:B)=H(A)+H(B)-H(A,B)$, with $H(Z)=-\sum_ZP_Z\log_2(P_Z)$.

The exact constraints that IC places on the strength of nonlocal correlations are, up to this day, unknown. However, this topic has received considerable attention, and several limitations in different nonlocality scenarios have been established \cite{recover,ahanj,ic_ml,xiang,mafalda}. In the following, we will combine theoretical considerations with the numerical characterization of $\tilde{Q}$ derived in Section \ref{sec_defin} to show that, in all such studies, $\tilde{Q}$ constitutes either the same or a better approximation to the quantum set.

Let us start with the original IC paper \cite{marcin}: there it is shown that, in bipartite scenarios with two inputs and two outputs, IC implies that the \emph{two-point correlators} $\langle A_xB_y\rangle\equiv P(a=b|x,y)-P(a\not=b|x,y)$ must satisfy Uffink's inequality \cite{uffink}:

\be
(\langle A_0B_0\rangle+\langle A_1B_0\rangle)^2+(\langle A_0B_1\rangle-\langle A_1B_1\rangle)^2\leq 4.
\ee

\noindent Now, two-point correlators arising from distributions compatible with ML have been shown to be compatible with quantum mechanics \cite{mac_loc}. The almost quantum set satisfies ML, and so it is also restricted by Uffink's inequality. Later works invoke nonlocality distillation arguments to strengthen the restrictions implied by the above inequality \cite{recover, ahanj}. However, since the almost quantum set is closed under wirings, it is obvious that it has to concur as well with such limitations, which sometimes differ from their quantum counterparts \cite{recover}. At least in the two-inputs/two-outputs scenario, $\tilde{Q}$ (or just ML) thus seems to enforce strictly stronger constraints on nonlocality than the principle of IC.

The restrictions stemming from IC have also been studied in many-outcome nonlocality scenarios, where it was shown that there exist ML correlations which would allow two parties to violate IC \cite{ic_ml}. More specifically, in \cite{ic_ml} the authors contemplate a setup where Alice (Bob) has $d$ (2) inputs and $d$ outputs and define the following generalization of a Popescu-Rorlich box \cite{popescu}:

\be
\mbox{PR}_0(a,b|x,y)=\left\{\begin{array}{l}1/d,\mbox {if } x\cdot y=b-a\mbox{ mod }d,\\0, \mbox{ otherwise}\end{array}\right.
\ee

\noindent Then, they consider boxes of the form $\mbox{PR}(E)=E\mbox{PR}_0+(1-E)\id$, where $\id(a,b|x,y)=\frac{1}{d}$ for all $x,y,a,b$. While $\mbox{PR}(\frac{1}{\sqrt{2}})$ seems to satisfy ML for all $d$, there exists a $d$-dependent value $E^{(d)}_{IC}$ such that, for $E>E^{(d)}_{IC}$, the box $\mbox{PR}(E)$ would allow two parties to violate IC. As it turns out, for $d=4,5$, $E^{(d)}_{IC}<\frac{1}{\sqrt{2}}$ \cite{ic_ml}. 

We used the SDP characterization of $\tilde{Q}$ in order to estimate the critical value $E^{(d)}_{\tilde{Q}}$ beyond which $\mbox{PR}(E)$ ceases to be almost quantum. Due to the size of the problem, for $d=4$ we used the modeling language YALMIP \cite{yalmip} to generate an input for the the semidefinite programming solver SDPT3 \cite{sdpt3}, which we run in the NEOS server \cite{neos}. The results are shown in Table 1.

\begin{table}
\begin{center}\begin{tabular}{|c|c|c|c|}
\hline
  $d$ &$E_{IC}$ & $E_{\tilde{Q}}$ &$E_{Q}$\\
  \hline
  $2$ & $0.707$ & $0.707$  & $0.707$\\
  $3$ & $0.708$ & $0.667$ & $0.667$\\
  $4$ & $0.705$ & $0.653$ & NA\\
  $5$ & $0.700$ & NA & $\geq 0.647$\\
  \hline
\end{tabular}
\caption{Maximal amount of nonlocality tolerated by IC, the almost quantum set $\tilde{Q}$ and quantum mechanics. For $d=5$, the SDP constraints were too numerous to be stored in memory in a normal desktop.}
\end{center}\label{ic_table}
\end{table}

As the reader can appreciate, $E_{IC}\geq E_{\tilde{Q}}$ for $d=2,3,4$. The vast amount of memory resources required to carry out optimizations over $\tilde{Q}$ in the $5255$ scenario prevented us from obtaining the value of $E_{\tilde{Q}}$ for $d=5$. Note, however, that the sequence of values $E^{(2)}_{\tilde{Q}},E^{(3)}_{\tilde{Q}},E^{(4)}_{\tilde{Q}}$ is strictly decreasing. It is therefore reasonable to venture that $E^{(5)}_{\tilde{Q}}<E^{(4)}_{\tilde{Q}}<E^{(5)}_{IC}$.

Moving on to the multipartite scenario, in \cite{xiang} the maximum quantum value of Svetlichny's Bell inequality \cite{svetlichny} is recovered from the principle of IC. Once more, the mechanism to show incompatibility with IC rests on wirings and the violation of Uffink's inequality, and therefore all points exhibiting a supraquantum violation of Svetlichny's inequality must also violate ML and hence lie outside $\tilde{Q}$. In \cite{mafalda}, the authors use the same tool to prove that the majority of the extreme points of the non-signalling polytope in the two-input/two-output tripartite scenario violates IC. More concretely, they show that all non-deterministic extreme points violate IC, except for one class, called $\sharp 4$ in \cite{tripartite}, which is provably compatible with IC \cite{rodrigo}. Since $\sharp 4$ violates LO as shown in \cite{loc_orth} and $\tilde{Q}$ satisfies LO this box is therefore not almost quantum. When just the extreme points of the tripartite no-signalling set in the simplest nonlocality scenario are considered, $\tilde{Q}$ thus enforces provably stronger constraints than IC.

\begin{figure}
  \centering
  \includegraphics[width=15 cm]{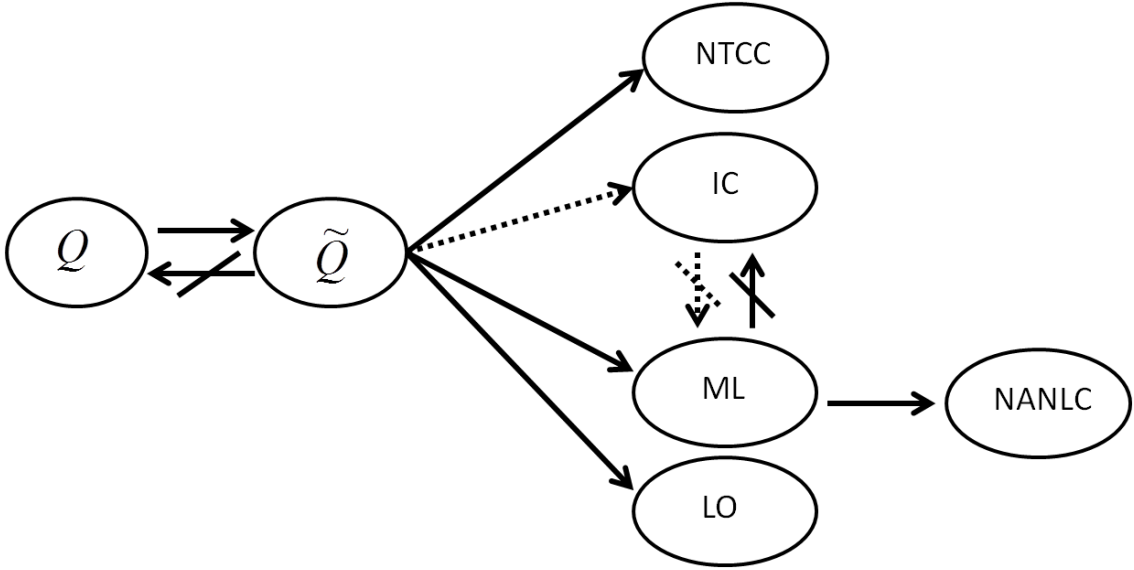}
  \caption{\textbf{Logical relations between physical principles, quantum nonlocality and the almost quantum set.} Solid arrows denote provable implications; dotted arrows represent implications for which so far there is only numerical evidence. That ML does not imply IC was proven in \cite{ic_ml}.}
  \label{esquema}
\end{figure}

\section{Conclusion}
\label{conclusion}
In this paper we have studied the set $\tilde{Q}$ of almost quantum correlations. This set appears naturally in a variety of seemingly unrelated fields, such as quantum information science, graph theory and quantum gravity. The ubiquity of the almost quantum set, together with the fact that $\tilde{Q}$ is closed under classical operations, seems to suggest that $\tilde{Q}$ emerges from a reasonable (yet unknown) physical theory. To support this conjecture, we have proven that almost quantum correlations satisfy a number of physical principles, originally conceived to single out the set of quantum correlations. The relations between these principles, quantum mechanics and $\tilde{Q}$ are summarized in Figure \ref{esquema}.

Note, however, that, despite our numerical evidence, we were not able to prove that $\tilde{Q}$ satisfies Information Causality \cite{marcin}. The original proof for the quantum case relies on the existence of a well-behaved entropic quantity, and so it does not carry through easily to the almost quantum case. Since the definition of sophisticated notions such as the von Neumann entropy requires the structure of a generalized probabilistic theory, finding a `natural' physical model whose non-locality is captured by $\tilde{Q}$ becomes imperative. Linking a physical theory with reasonable entropic inequalities to the almost quantum set would prove that $\tilde{Q}$ not only respects Information Causality, but also any future information-theoretic principle derived from, say, strong subaddititity, or the data processing inequality. In addition, an explicit `almost quantum theory' would also suggest where to look for genuinely non-quantum behavior and thus could be the first step towards an experimental refutation of quantum theory.

\section*{Acknowledgements}
We thank Stefano Pironio and Tam\'as V\'ertesi for providing numerical examples of non-quantum distributions in $\tilde{Q}$, and acknowledge interesting discussions with Joe Henson. This work is supported by the ERC AdG NLST and CoG QITBOX, the Chist-Era project DIQIP grant and the project ``Intrinsic randomness in the quantum
world" from the John Templeton Foundation.

\begin{appendix}

\section{Proof of Lemma \ref{quasi_nielsen}}
\label{non_triv_ap}

$P(\bar{a},b|\bar{x},[\bar{z}, \bar{y}])\in \tilde{Q}$ implies that there exists a pure quantum state $\ket{\phi}$ and projector operators $\{E^{\bar{a},\bar{x}},F^{b,[\bar{z},\bar{y}]}\}$ such that

\begin{enumerate}
\item $P(\bar{a},b|\bar{x},[\bar{z}, \bar{y}])=\bra{\phi}E^{\bar{a},\bar{x}}F^{b,[\bar{z},\bar{y}]}\ket{\phi}$.

\item $\sum_{\bar{a}}E^{\bar{a},\bar{x}}=\sum_bF^{b,[\bar{z},\bar{y}]}=\id$.

\item $E^{\bar{a},\bar{x}}F^{b,[\bar{z},\bar{y}]}\ket{\phi}=F^{b,[\bar{z},\bar{y}]}E^{\bar{a},\bar{x}}\ket{\phi}$, for all $\bar{x},\bar{y},\bar{z},\bar{a},b$.

\end{enumerate}

Now, consider the distribution 

\be
P'(\bar{a},\bar{c}|\bar{x},\bar{z})\equiv\tr\left\{(E^{\bar{a},\bar{x}}\otimes \id_2^{\otimes n+1})\Pi^{\bar{c},\bar{z}} \rho\right\},
\label{ns_cons}
\ee

\noindent where $\rho=\proj{\phi}\otimes \proj{\psi_{0}}^{\otimes n}\otimes\proj{\psi_{1}}$, and

\be
\Pi^{\bar{c},\bar{z}}=U^\dagger_{\bar{z}}(\id\otimes \proj{\psi_{\bar{c}}}\otimes \id_2)U_{\bar{z}},
\ee

\noindent with

\be
U_{\bar{z}}=\sum_{\bar{y},b}F^{b,[\bar{z},\bar{y}]}\otimes \proj{\bar{y}}\otimes \sigma^b.
\ee

\noindent Here $\ket{\bar{y}}$ denotes the expression of the bit string $\bar{y}$ in the computational basis, i.e., $\ket{\bar{y}}\equiv\bigotimes_{i=1}^n\ket{y_i}$. Analogously, $\ket{\psi_{\bar{c}}}$ represents the expression of $\bar{c}$ in the Hadamard basis; $\ket{\psi_{\bar{c}}}\equiv\bigotimes_{i=1}^n\frac{1}{\sqrt{2}}(\ket{0}+(-1)^{c_i}\ket{1})$. $\sigma$ denotes the Pauli matrix $\ket{0}\bra{1}+\ket{1}\bra{0}$.

One can verify that $\{U_{\bar{z}}\}$ are unitary operators. It follows that, for any $\bar{z}\in \{0,1\}^m$, $\{\Pi^{\bar{c},\bar{z}}\}_{\bar{c}}$ is a complete set of projector operators. These two features imply that $P'(\bar{a},\bar{c}|\bar{x},\bar{z})$ is normalized. $P'(\bar{a},\bar{c}|\bar{x},\bar{z})$ also satisfies the no-signalling conditions:

\be
\sum_{\bar{a}}P'(\bar{a},\bar{c}|\bar{x},\bar{z})=P'(\bar{c}|\bar{z}),\qquad\sum_{\bar{c}}P'(\bar{a},\bar{c}|\bar{x},\bar{z})=P'(\bar{a}|\bar{x}).
\ee

\noindent To prove that $P'(\bar{a},\bar{c}|\bar{x},\bar{z})$ is a non-signalling distribution, it remains to be seen that $P'(\bar{a},\bar{c}|\bar{x},\bar{z})\geq 0$ for all inputs and outputs $\bar{a}, \bar{c}, \bar{x}, \bar{z}$. Expanding eq. (\ref{ns_cons}) we have that

\begin{align}
P'(\bar{a},\bar{c}|\bar{x},\bar{z})&\equiv\tr\left\{(E^{\bar{a},\bar{x}}\otimes \id_2^{\otimes n+1})\Pi^{\bar{c},\bar{z}} \rho\right\}\\
&= \tr\left\{(E^{\bar{a},\bar{x}}\otimes \id_2^{\otimes n+1}) U_{\bar{z}}^\dagger \bigg( \mathbb{I}\otimes \proj{\psi_{\bar{c}}}\otimes \mathbb{I}_2\bigg) U_{\bar{z}}\rho\right\}\\
&= \tr\left\{\sum_{\bar{y},b}\rho(E^{\bar{a},\bar{x}}\otimes \id_2^{\otimes n+1})\bigg(F^{b,[\bar{z},\bar{y}]}\otimes \proj{\bar{y}}\otimes \sigma^{b} \bigg)\bigg( \mathbb{I}\otimes \proj{\psi_{\bar{c}}}\otimes \mathbb{I}_2\bigg) U_{\bar{z}}\right\}.
\end{align}
Notice that, due to point 3, the operators $E^{\bar{a},\bar{x}}$ and $F^{b,[\bar{z},\bar{y}]}$ can be interchanged, since they act on the $\bra{\phi}$ subspace of the density operator $\rho$. Thus we have 

\be
P'(\bar{a},\bar{c}|\bar{x},\bar{z})=\tr(\rho U^\dagger_{\bar{z}}(E^{\bar{a},\bar{x}}\otimes \proj{\psi_{\bar{c}}}\otimes \id_2)U_{\bar{z}})\geq 0.
\ee

Finally, we must show that eq. (\ref{robustness}) holds. Note that

\be
\Pi^{\bar{c},\bar{z}}=\frac{1}{2^{n}}\sum_{\bar{y},b,\bar{y}',b'}(-1)^{\bar{c}\cdot(\bar{y}+\bar{y}')}F^{b,[\bar{z},\bar{y}]}F^{b',[\bar{z},\bar{y}']}\otimes\ket{\bar{y}}\bra{\bar{y}'}\otimes \sigma^{b+b'},
\ee
\noindent where the factor $ \frac{1}{2^n}$ arises from the overlap between the Hadamard and computational basis. We thus have that

\begin{eqnarray}
\sum_{\bar{a}}P'(\bar{a},\bar{x}|\bar{x},\bar{a})&=&\tr\left\{\rho\sum_{\bar{a}} (E^{\bar{a},\bar{x}}\otimes \id_2^{\otimes n+1})\Pi^{\bar{x},\bar{a}}\right\}=\nonumber\\
&=&\frac{1}{2^{2n}}\sum_{\bar{a},\bar{y},\bar{y}',b,b'}\bra{\phi}E^{\bar{a},\bar{x}} F^{b,[\bar{a},\bar{y}]}F^{b',[\bar{a},\bar{y}']}\ket{\phi}(-1)^{b+b'+\bar{x}\cdot(\bar{y}+\bar{y}')}.
\label{inter}
\end{eqnarray}

\noindent From property (3) and the fact that $(E^{\bar{a},\bar{x}})^2=E^{\bar{a},\bar{x}}$, we have that 

\be
\bra{\phi}(E^{\bar{a},\bar{x}} F^{b,[\bar{a},\bar{y}]})\cdot F^{b',[\bar{a},\bar{y}']}\ket{\phi}=\bra{\phi} (F^{b,[\bar{a},\bar{y}]}E^{\bar{a},\bar{x}})\cdot (E^{\bar{a},\bar{x}} F^{b',[\bar{a},\bar{y}']})\ket{\phi}.
\ee

\noindent Defining

\be
\ket{\bar{a},\bar{x}}\equiv E^{\bar{a},\bar{x}} \ket{\phi},\ket{\bar{a},b,\bar{x},[\bar{z},\bar{y}]}\equiv E^{\bar{a},\bar{x}} F^{b,[\bar{z},\bar{y}]}\ket{\phi},
\ee

\noindent we can therefore re-express eq. (\ref{inter}) as

\be
\sum_{\bar{a}}P'(\bar{a},\bar{x}|\bar{x},\bar{a})=\sum_{\bar{a}}\braket{V^{\bar{a},\bar{x}}}{V^{\bar{a},\bar{x}}},
\label{inter2}
\ee

\noindent with

\be
\ket{V^{\bar{a},\bar{x}}}\equiv\frac{1}{2^n}\sum_{\bar{y},b}(-1)^{b+\bar{x}\cdot\bar{y}}\ket{\bar{a},b,\bar{x},[\bar{a},\bar{y}]}.
\label{V}
\ee

\noindent Now, notice that $\braket{\bar{a},\bar{x}}{\bar{a},\bar{x}}=\bra{\phi}E^{\bar{a},\bar{x}}\ket{\phi}=P(\bar{a}|\bar{x})$. Likewise,

\be
\braket{\bar{a},b,\bar{x},[\bar{z},\bar{y}]}{\bar{a},b,\bar{x},[\bar{z},\bar{y}]}=\braket{\bar{a},b,\bar{x},[\bar{z},\bar{y}]}{\bar{a},\bar{x}}=\bra{\phi}F^{b,[\bar{z},\bar{y}]} E^{\bar{a},\bar{x}}\ket{\phi}=P(\bar{a},b|\bar{x},[\bar{z},\bar{y}]).
\ee

\noindent These two relations imply that

\be
\ket{\bar{a},b,\bar{x},[\bar{z},\bar{y}]}=P(b|\bar{x},\bar{a},[\bar{z},\bar{y}])\ket{\bar{a},\bar{x}}+\ket{\perp},
\label{perpe}
\ee

\noindent for some vector $\ket{\perp}$ with $\braket{\perp}{\bar{a},\bar{x}}=0$.

\noindent Substituting (\ref{perpe}) in (\ref{V}), we find that

\be
\ket{V^{\bar{a},\bar{x}}}=\left\{2\sum_{\bar{y}}\frac{P(b=\bar{x}\cdot\bar{y}|\bar{x},\bar{a},[\bar{a},\bar{y}])}{2^n}-1\right\}\ket{\bar{a},\bar{x}}+\ket{\perp}'.
\ee

\noindent Neglecting the contribution of $\ket{\perp}'$ to the norm of $\ket{V^{\bar{a},\bar{x}}}$, and invoking once more the identity $\braket{\bar{a},\bar{x}}{\bar{a},\bar{x}}=P(\bar{a}|\bar{x})$, we have that the right-hand side of eq. (\ref{inter2}) is lowerbounded by eq. (\ref{robustness}).

\end{appendix}

\end{document}